\documentclass[11pt]{article}

\usepackage{cite}
\usepackage{fullpage}
\usepackage{amssymb}
\usepackage{amsmath}
\usepackage{graphicx}
\usepackage{enumerate}
\usepackage{multirow}

\newcommand{\str}{\ensuremath{S} }
\newcommand{\tree}{\ensuremath{T}}
\newcommand{\slp}{\ensuremath{\mathcal{S}} }

\newcommand{\fca}{\ensuremath{\textsc{firstcolor}}}
\newcommand{\lca}{\ensuremath{\textsc{lastcolor}}}
\newcommand{\labsuc}{\ensuremath{\textsc{ls}}}
\newcommand{\la}{\ensuremath{\textsc{la}}}
\newcommand{\leftc}{\ensuremath{\text{\textit{left}}}}
\newcommand{\rightc}{\ensuremath{\text{\textit{right}}}}
\newcommand{\band}{\ensuremath{\wedge}}
\newcommand{\bor}{\ensuremath{\vee}}
\newcommand{\bxor}{\ensuremath{\oplus}}
\newcommand{\hf}{\ensuremath{\mathcal{H}}}
\newcommand{\pat}{\ensuremath{P}}

\newtheorem{lemma}{Lemma}
\newtheorem{theorem}{Theorem}

\newcommand{\qed}{\hfill\ensuremath{\Box}\medskip\\\noindent}
\newenvironment{proof}{\noindent\emph{Proof. }}

\title{Compressed Subsequence Matching and Packed~Tree~Coloring}
\author{Philip Bille\thanks{Supported in part by the The Danish Council for Independent Research $\vert$ Natural Sciences grant DFF ­1323--00178.} \and Patrick Hagge Cording\\[0.4em]
\texttt{\small{\{phbi,phaco,inge\}@dtu.dk}} \and Inge Li G\o rtz\footnotemark[1]}

\begin{document}
	
\maketitle

\begin{abstract}
\noindent We present a new algorithm for subsequence matching in grammar compressed strings. Given a grammar of size $n$ compressing a string of size $N$ and a pattern string of size $m$ over an alphabet of size $\sigma$, our algorithm uses $O(n+\frac{n\sigma}{w})$ space and $O(n+\frac{n\sigma}{w}+m\log N\log w\cdot occ)$ or $O(n+\frac{n\sigma}{w}\log w+m\log N\cdot occ)$ time. Here $w$ is the word size and $occ$ is the number of occurrences of the pattern. Our algorithm uses less space than previous algorithms and is also faster for $occ=o(\frac{n}{\log N})$ occurrences. The algorithm uses a new data structure that allows us to efficiently find the next occurrence of a given character after a given position in a compressed string. This data structure in turn is based on a new data structure for the tree color problem, where the node colors are packed in bit strings.
\end{abstract}

\section{Introduction}
% \begin{itemize}
% 	\item Given an SLP \slp of size $n$ compressing a string \str of size $N$ and a pattern string $P$ of size $m$ over an alphabet of size $\sigma$, the \textit{compressed subsequence matching problem} is to find and report the index of all minimal substrings of $\str$ that contain $P$ as a subsequence. A substring is said to be minimal if shortening it implies that $P$ is no longer a subsequence of that substring.
% 	\item In this paper we present a new algorithm for compressed subsequence matching which uses less space than previously known algorithms and is faster than the  previously fastest algorithm for a bounded number of occurrences.
% 	\item Our algorithm breaks with the table filling approach common to previous algorithms 
% \end{itemize}

In the  \textit{compressed subsequence matching problem} we are given a grammar \slp of size $n$ compressing a string \str of size $N$ and a pattern string $P$ of size $m$ over an alphabet of size $\sigma$, and the goal is to find and report the index of all minimal substrings of $\str$ that contain $P$ as a subsequence. A substring is said to be minimal if shortening it implies that $P$ is no longer a subsequence of that substring. In this paper we present a new algorithm for compressed subsequence matching which is space efficient and is faster than the previously fastest algorithm for a bounded number of occurrences. Our algorithm relies on a method that is different from the ones used by previous algorithms.

% \begin{itemize}
% 		\item Log data is sequential in time and can be seen as a string
% 		\item A common task is to search the log data for a sequence of events, say $A$ followed by $B$ followed by $C$, where the events may be separated by other events
% 		\item Many applications will automatically compress log data to save space
% 		\item The bottleneck when searching the compressed data is to decompress it
% 		\item Subsequence matching was also considered in relation to knowledge discovery and data mining \cite{mannila1997discovery}
% 		\item Processing compressed strings without full decompression is a natural goal 
% 		\item Subsequence matching may be seen as a restricted variant of approximate string matching where only deletions in the subject string (analogously, insertions in the pattern) is allowed
% \end{itemize}

Subsequence matching is useful when searching sequential log data for a sequence of events that may be separated by other events. Say for instance that we are running a webserver and we want to know how often a visitor has found her way to subpage $C$ through page $A$ and then $B$. We then set $\pat = ABC$ and apply a subsequence matching algorithm to the contents of the log file. Many applications will automatically compress log data to save space, and so the bottleneck of the procedure becomes decompression of the data. In this case, processing the data without fully decompressing it, is crucial. Subsequence matching was also considered in relation to knowledge discovery and data mining \cite{mannila1997discovery}.

	% \begin{itemize}
	% 	\item For uncompressed strings several algorithms have been presented \cite{mannila1997discovery,das1997episode,boasson1999window,cegielski2006multiple,baeza1991searching,tronicek2001episode,crochemore2003directed}
	% 	\item We may apply the online algorithm due to Das et al. to get an algorithm that uses $O(\frac{Nm}{\log m})$ time and $O(n)$ space
	% 	\item Cegielski et al., 2006 \cite{cegielski2006window}: Invents the table algorithm, and the tables are computed in $O(nm^2\log m)$ time and $O(nm^2)$ space
	% 	\item Tiskin, 2009 \cite{tiskin2009faster}: $O(nm^{1.5})$
	% 	\item Tiskin, 2011: $O(nm\log m)$ time and space(?) .. see http://arxiv.org/pdf/0707.3619.pdf
	% 	\item Yamamoto et al., 2011 \cite{yamamoto2011faster}: Improves Cegielski's algorithm to run in $O(nm)$ time and space
	% \end{itemize}

Several algorithms have been presented for uncompressed strings \cite{mannila1997discovery,das1997episode,boasson1999window,cegielski2006multiple,baeza1991searching,tronicek2001episode,crochemore2003directed}. The fastest of these is due to Das et al.~\cite{das1997episode}. Since it is an online algorithm we may apply it to the compressed version without having to store the entire decompressed string, and we get an algorithm with running time $O(\frac{Nm}{\log m})$ that uses $O(n+m)$ space. The first algorithm with time complexity independent of the size of the string was presented by Cegielski~et~al.~\cite{cegielski2006window} in 2006. Its runnning time is $O(nm^2\log m+occ)$ time and it uses $O(nm^2)$ space. Using a different approach, Tiskin improved the running time to $O(nm^{1.5}+occ)$ \cite{tiskin2009faster} and later even further to $O(nm\log m+occ)$~\cite{tiskin2011towards}. The space usage of his algorithms is $O(nm)$. The most recent improvement is due to Yamamoto~et~al.~\cite{yamamoto2011faster} who present an algorithm based on the ideas of Cegielski~et~al. that runs in $O(nm+occ)$ time and $O(nm)$ space. All results are summarized in Table 1.

\begin{table}[h]\label{tab:results}
\begin{center}
\begin{tabular}{lll}
	\hline
	Time complexity & Space complexity & Author(s) \\
	\hline
	$O(\frac{Nm}{\log m})$ & $O(n+m)$ & Das et al.~\cite{das1997episode} \\
	$O(nm^2\log m+occ)$ & $O(nm^2)$ & Cegielski~et~al.~\cite{cegielski2006window} \\
	$O(nm^{1.5}+occ)$ & $O(nm)$ & Tiskin~\cite{tiskin2009faster} \\
	$O(nm\log m+occ)$ & $O(nm)$ & Tiskin~\cite{tiskin2011towards} \\
	$O(nm+occ)$ & $O(nm)$ & Yamamoto~et~al.~\cite{yamamoto2011faster} \\
	\hline
	$O(n+\frac{n\sigma}{w}+m\log N\log w\cdot \text{occ})$ & \multirow{2}{*}{$O(n+\frac{n\sigma}{w})$} & \multirow{2}{*}{This paper} \\	
	$O(n+\frac{n\sigma}{w}\log w+m\log N\cdot \text{occ})$ & & \\
	\hline
\end{tabular}
\caption{Time and space complexities of algorithms for compressed subsequence matching.}
\end{center}
\end{table}

\noindent Assume without loss of generality that the compressed string is given as a Straight Line Program (SLP). An SLP is an acyclic grammar in Chomsky normal form, i.e., a grammar where each nonterminal production rule expands to two other rules and  generates one string only. SLPs are widely studied because they model many well-known compression schemes, such as LZ77 \cite{lz77}, LZ78 \cite{lz78}, and Re-Pair \cite{larsson2000off}  with little overhead \cite{charikar2005smallest,rytter2003application}. The following theorem is the main result of this work.

%\cite{bille2013fingerprints,charikar2005smallest,rytter2003application}

\begin{theorem}\label{thm:SCSM}
Given an SLP $\slp$ of size $n$ compressing a string $\str$ of size $N$ and a pattern $\pat$ of size $m$ over an alphabet of size $\sigma$, compressed subsequence matching can be solved in $O(n+\frac{n\sigma}{w})$ words of space and time
\begin{enumerate}
\item[\textit{(i)}] $O(n+\frac{n\sigma}{w}+m\log N\log w\cdot \text{occ})$, or
\item[\textit{(ii)}] $O(n+\frac{n\sigma}{w}\log w+m\log N\cdot \text{occ})$
\end{enumerate}
in the word RAM model with word size $w\geq \log N$, and where $\text{occ}$ is the number of minimal occurrences of $\pat$ in $\str$.
\end{theorem}

\noindent Our new algorithm uses less space (linear in $n$ if $\sigma\leq w$) and is also faster than the previously fastest algorithm for $o(\frac{n}{\log N})$ occurrences when $\sigma \leq m$. Note that we can guarantee that the latter requirement always holds by bounding $\sigma=O(m)$ using hashing in return for using $O(m)$ additional extra space.

The algorithm is based on the idea of a simple algorithm for subsequence matching in uncompressed strings which basically scans the string for occurrences of the pattern. We speed up the scanning on compressed strings by introducing the first data structure for SLPs that supports labelled successor queries. The answer to a labelled succesor query $\labsuc(i,c)$ on a string is the index of the first character $c$ occurring after position $i$ in the string. An essential part of this data structure is a new data structure for the tree color problem. This problem is to preprocess a tree where each node is colored by zero or more colors, such that given a node $v$ and a color $c$, we may efficiently answer a first colored ancestor query, i.e., compute the lowest ancestor of $v$ with color $c$. Additionally, this data structure also supports a new type of query we call the last colored ancestor. Here the query is two nodes $u$ and $v$ and a color $c$, and the answer is the highest node on the path from $u$ to $v$ with color $c$. These results may be of independent interest.

% \begin{itemize}
% 	\item The new algorithm uses less space (linear in $n$ if $\sigma\leq w$) and is also faster than the previously fastest algorithm for $o(\frac{n}{\log N})$ occurrences
% 	\item The algorithm is based on the idea of a simple algorithm for subsequence matching in uncompressed strings which basically scans the string for occurrences of the pattern
% 	\item We speed up the scanning on compressed strings by introducing the first data structure for SLPs that supports labelled successor queries
% 	\item The answer to a labelled succesor query $\labsuc(i,c)$ on a string is the index of the first character $c$ occurring after $i$ in the string
% 	\item In doing so, we present a new data structure for the tree color problem
% 	\item The tree color problem is to preprocess a tree where each node has zero or more colors such that given a node $v$ and a color $c$, we may efficiently compute the lowest ancestor of $v$ with color $c$
% \end{itemize}

% \begin{itemize}
% 	\item This paper is organized such that it describes the data structure needed for Theorem~\ref{thm:SCSM} in a bottom up approach
% 	\item That is, we start by describing a our new result for the tree color problem, then move on to the labelled successor data structure, and ultimately describe the algorithm for subsequence matching
% \end{itemize}

This paper is organized such that we start by describing our new result for the tree color problem (after a section of preliminaries), then move on to the labelled successor data structure, and finally describe the algorithm for subsequence matching.

\section{Preliminaries}

%\subsection{Bit Strings}

\paragraph{Bit Strings.}
We will use bit strings to represent sets. In a bit string $B=b_1b_2\ldots b_u$ representing a set $\mathcal{B}$ of elements from a universe of size $u$, $b_i=\mathtt{1}$ iff element $i$ is in $\mathcal{B}$. $B=[\mathtt{0}]^u$ denotes the empty set. The operators $\band$, $\bor$, and $\bxor$ denote the bitwise AND, OR, and exclusive OR (XOR) of two bit strings. The negation of a bit string $B$ is $\overline{B}$. A \emph{summary} $B_s$ of $k$ bit strings $B_1,B_2, \ldots,B_k$ of equal length is $B_s=B_1 \bor B_2 \bor \ldots \bor B_k$. For a bit string of length $w$ we assume that the mask of any constant can be computed in $O(1)$ time. Given a bit string $B=b_1b_2\ldots b_w$, $b_1$ is the most significant bit. The index of the least significant set bit can be found in $O(1)$ time from $\log_2 ( \overline{(B-1)\bxor B}\band B )$. Finding the most significant set bit is more elaborate, but can also be done $O(1)$ time~\cite{fredman1993surpassing}. An $n\times m$ bit matrix may be transposed in $O(w\log w)$ time if $n\leq w$ and $m\leq w$~\cite{thorup2002randomized}.

\paragraph{Trees.}
In this paper all trees are rooted, ordered, and have labels on the nodes. The number of nodes in a tree $\tree$ is $t$. We denote by $\tree(v)$ the subtree rooted in $v$ containing all descendants of $v$. The size $|\tree(v)|$ is the number of nodes in the subtree $\tree(v)$ including $v$. If $u$ is a node in the subtree $\tree(v)$ we write $u\in \tree(v)$. If $\tree$ is a binary tree we denote the left and right child of a node $v$ by $\leftc(v)$ and $\rightc(v)$. %The following is a selection of results and techniques for trees that we will use in this paper.

%\paragraph{Heavy path decomposition.} 
A heavy path decomposition \cite{sleator1983data} decomposes $\tree$ into disjoint paths. Nodes are classified as either heavy or light and the decomposition is defined as follows. The root is light. For each internal node $v$, its heavy child $w$ is the node for which $\tree(w)$ is of maximum size among the subtrees rooted in children of $v$. The other children of $v$ are light. Edges are also classified as heavy and light. An edge going into a heavy node is heavy and likewise for light nodes. The heavy path decomposition ensures the property that $\frac{1}{2}|\tree(v)|>|\tree(u)|$ for any light child $u$ of $v$. This means that there are $O(\log t)$ light edges on any path from the root to a leaf. The heavy path decomposition can be computed in $O(t)$ time and space.

%\paragraph{Cluster partition.} 
Given a binary tree $\tree$ rooted in a node $r$, $t>1$, and a parameter $ 1 \leq x \leq t$, we may partition $\tree$ into at most $t/x$ clusters such that for a fixed constant $c$, the size of any cluster is at most $cx$ \cite{alstrup1997optimal,alstrup1997minimizing} (see also \cite{abiteboul2006compact} for a full proof). Two clusters overlap in at most one node, and a node is called a boundary node if it is part of more than one cluster. Any cluster has at most two boundary nodes, and a boundary node is either a leaf or the root in the subtree that is the cluster. The tree obtained by repeatedly contracting edges between two nodes if one of them is not a boundary node is called the macro tree. In other words, the macro tree is the tree consisting only of boundary nodes. A cluster partition can be found in $O(t)$ time.

	%\paragraph{Level ancestor.} 
The answer to a level ancestor query $\la(v,d)$ on $\tree$ is the ancestor of $v$ with depth $d$. A linear space data structure that answers an $\la$ query in $O(1)$ time can be computed for $\tree$ in $O(t)$ time~\cite{dietz1991finding} (see also \cite{berkman1994finding,alstrup2000improved,bender2004level}).

%\subsection{Straight Line Programs}
\paragraph{Straight Line Programs.}
A Straight Line Program \slp is a context-free grammar in Chomsky normal form with $n$ production rules that unambigously derives a string $\str$ of length $N$. We represent the SLP as a rooted, ordered, and node-labelled directed acyclic graph (DAG) with outdegree $2$ and we will refer to production rules as nodes in the DAG. A depth-first left-to-right traversal starting from a node $v$ in the DAG produces the string $\str(v)$ of length $|\str(v)|$. The tree that emerges from the traversal we call the derivation tree. We denote the left and right children of $v$ for $\leftc(v)$ and $\rightc(v)$, respectively. Furthermore, the height of the SLP is the length of the longest path going from the root to a terminal node and is denoted by $h$.

% A non-terminal production rule has two variables on its right side and a terminal a single character from the alphabet $\Sigma_\slp$ on its right side.

%\paragraph{Random access in SLPs.} 
We may access a character $\str[i]$ in $O(h)$ time by storing $|S(v)|$ for each node $v$ in the SLP, and simulate a top-down search of the derivation tree. Doing so yields a unique path from the root of $\slp$ to the terminal node labelled $\str[i]$. There is also a linear space data structure that supports random access in SLPs in $O(\log N)$ time~\cite{bille2011random}. A key technique used in this data structure is the extension of the heavy path decomposition of trees to SLPs which we will also use in our data structure. For each node $v \in \slp$, we select the child of $v$ that derives the longest string to be a heavy node. The other child is light. Heavy and light edges are defined as in the decomposition of trees. Whereas applying this technique to a tree results in a decomposition into disjoint paths, it will result in a decomposition into disjoint trees when applied to an SLP. We denote this set of trees by the heavy forest $\hf$ of the SLP. This decomposition ensures that the number of light edges on any path from the root to a terminal node is $O(\log N)$. Hence, on any path from the root of the SLP to a terminal node, we visit at most $\log N$ trees from $\hf$. When accessing a character using the data structure of \cite{bille2011random} we may also report the entry and exit nodes for each tree visited on the unique root-to-terminal path that emerges from the query.

\section{Packed Tree Color Problems}
In a colored tree, each node is colored by zero or more colors from the set $\{1,\ldots,\sigma\}$. A packed colored tree is a colored tree where the colors of each node $v$ is given as a bit string $C(v)$ where $C(v)[c]=1$ iff $v$ is colored $c$. In this section we consider the \textit{packed tree color problem} which is to preprocess a packed colored tree $\tree$ to support  first and last colored ancestor queries. The answer to a first colored ancestor query $\fca(v,c)$ is the lowest ancestor of $v$ with color $c$, and the answer to a last colored ancestor query $\lca(u,v,c)$ is the highest node with color $c$ on the path from $u$ to $v$, where we always assume that $u$ is an ancestor of $v$. Throughout this section we will use the following notation to distinguish results. If a data structure requires $p(t)$ time to build, uses $s(t)$ space, and supports $\fca$ and $\lca$ queries in $q(t)$ time, then the the triple $\langle p(t), s(t), q(t)\rangle$ refers to the solution.

Solutions to the tree color problem for trees that are not packed may be applied to packed trees. All known solutions focus entirely on supporting $\fca$ queries \cite{dietz1991finding,muthukrishnan1996time,ferragina1996efficient,alstrup1998marked}. A simple solution that supports $\fca$ queries in $O(1)$ time is to store the answer for every color in every node. This yields a $\langle O(t\sigma),O(t\sigma),O(1)\rangle$ solution. The currently best known trade-off for the tree color problem is $\langle O(t+D),O(t+D),O(\log w)\rangle$  \cite{muthukrishnan1996time}, where $D=\sum_{v\in T}\sum_{i=1}^\sigma C(v)[i]$ is the accumulated number of colors used. %As a building block for the random access data structure developed by Bille et al.~\cite{bille2011random}, a solution to the tree color problem is also given. It has  $O(1)$ query time and $o(n\sigma)$ preprocessing time and space, but it comes with the restriction that the nodes can have only one color. 

Our motivation for revisiting this problem is twofold. First we have that $D=O(t\sigma)$ in our application and we are striving for a space bound that is in $o(t\sigma)$. Second we want to support $\lca$ queries. %We first present three solutions and then combine them to a fourth with a new and desireable time-space trade-off.

In this section we present three solutions to the packed tree coloring problem and combine them to a data structure with a new and desireable time-space trade-off.

%Due to lack of space, the analyses of the first two data structures are omitted.

\subsection{A $\langle O(t\sigma), O(t\sigma), O(1) \rangle$ Solution}
%For this data structure we show how a $\lca$ query can be reduced to two $\fca$ queries and one $\la$ query if we allow the data structure to use $O(t\sigma)$ space. The main idea comes from \cite{muthukrishnan1996time}.

%\paragraph{Data structure.} 
We store the result of a $\fca(v,c)$ query for every node and color. For each color, let the induced $c$-colored subtree be the tree obtained by deleting all nodes that are not colored by color $c$ except the root. Build a levelled ancestor data structure for each induced colored subtree.

%\paragraph{Query.} 
The result of a $\fca$ query is precomputed. A $\lca(u,v,c)$ query is answered as follows. If $\fca(v,c)=\fca(u,c)$ then there is not a node with color $c$ on the path from $u$ to $v$. If $\fca(v,c)\neq \fca(u,c)$ then let $v'$ and $u'$ be the nodes corresponding to $\fca(v,c)$ and $\fca(u,c)$ in the induced $c$-colored subtree. The answer to $\lca(u,v,c)$ is then the answer to $\la(v',depth(u')-1)$ in the induced $c$-colored subtree.

%\paragraph{Analysis.} 
The results of $\fca$ queries can be found and stored using $O(t\sigma)$ time and space. The induced colored subtrees can be computed in $O(t\sigma)$ time and use $O(D)=O(t\sigma)$ space. A $\fca$ query clearly takes $O(1)$ time. For a $\lca$ query, we perform two $\fca$ queries and one $\la$ query, each of which takes constant time.

\begin{lemma}
The packed tree color problem can be solved using $O(t\sigma)$ preprocessing time and space, and $O(1)$ query time.
\end{lemma}
%\begin{proof}
% The results of $\fca$ queries can be found and stored using $O(t\sigma)$ time and space. The induced colored subtrees can be computed in $O(t\sigma)$ time and use $O(D)=O(t\sigma)$ space. A $\fca$ query clearly takes $O(1)$ time. For a $\lca$ query, we perform two $\fca$ queries and one $\la$ query, each of which takes constant time.\qed
% \end{proof}

\subsection{A $\langle O(t+\frac{t\sigma}{w}),O(t+\frac{t\sigma}{w}),O(\log t) \rangle$ Solution}

%This data structure decomposes $\tree$ into heavy paths and we build a binary tree of summaries for each path.

%\paragraph{Data structure.} 
We fix a heavy path decomposition of $\tree$. For each path $p$ in the heavy path decomposition of $\tree$ we build a balanced binary tree $\tree_p$ having the nodes of $p$ as leaves. For each node $v$ in $\tree_p$ we store a summary $B(v)$ of the colors of its children. For each heavy path $p=v_1,v_2,\ldots,v_k$, where $v_1$ is the highest node on the path, we store a summary $P(v_i)$ of colors on the path prefix $v_1\ldots v_i$ for every $v_i$ on $p$. 

%\paragraph{Query.}
For answering a $\fca(v,c)$ query, let $p=v_1,v_2,\ldots,v_k$ be the heavy path containing $v$ and let $v_i=v$ for some $1\leq i \leq k$. If $P(v_i)[c]=\mathtt{1}$ we find the lowest ancestor $v_a$ of $v_i$ in $\tree_p$ for which $B(\leftc(v_a))[c]=\mathtt{1}$ and $v_i\notin \tree_p(\leftc(v_a))$. The answer to the query is then the rightmost leaf in $\tree_p(\leftc(v_a))$ with color $c$. If $P(v_i)[c]=\mathtt{0}$ we repeat the procedure with $v_i=parent(v_1)$, i.e., we jump to the previous heavy path, until we find the first colored ancestor or we reach the root of $\tree$.

A $\lca(u,v,c)$ query is handled in a similar way. We first find the highest light node $w$ on the path from $u$ to $v$ for which $P(parent(w))[c]=\mathtt{1}$. Let $p$ be the heavy path containing $parent(w)$. Now there are three cases. If $u$ is not on $p$, the answer to the query is the leftmost leaf in $T_p$ that has color $c$. If $p$ contains $u$, the answer is the leftmost leaf with color $c$ to the right of $u$ in $\tree_p$, if such a node exists. If it does not exist, we repeat the first step for the second highest light node $w'$ between $u$ and $v$ for which $P(parent(w'))[c]=\mathtt{1}$.

%\paragraph{Analysis.} 
The heavy path decomposition of $\tree$ can be found and stored in $O(t)$ time and space. Since the paths of the heavy path decomposition are disjoint, the total number of leaves in the binary summary trees is $t$, so the total number of nodes in the trees is $O(t)$. We store $O(t)$ summary bit vectors of size $O(\frac{\sigma}{w})$ using a total of $O(\frac{t\sigma}{w})$ space. We use $O(\frac{t\sigma}{w})$ bitwise OR operations to create the summaries in a bottom up fashion. In total, preprocessing time and space usage is $O(t+\frac{t\sigma}{w})$.

For both queries we visit at most $\log t$ heavy paths. When the path with the answer has been found we walk up the binary tree and then down again. Since the tree is balanced and has at most $t$ leaves, this takes $O(\log t)$ time. For $\lca$ queries we do this at most twice. The query time for $\fca$ and $\lca$ queries is therefore $O(\log t)$ time.

\begin{lemma}
The packed tree color problem can be solved using $O(t+\frac{t\sigma}{w})$ preprocessing time and space, and $O(\log t)$ query time.
\end{lemma}

\subsection{A $\langle O(t+\frac{t\sigma\log w}{w}),O(t+\frac{t\sigma}{w}),O(\frac t w) \rangle$ Solution}

% Intro

%\paragraph{Data structure.} 
Let $v_1,\ldots,v_t$ be the nodes of $\tree$ in pre-order. We will represent $\tree$ as a $\sigma \times t$ bit matrix $M$. Let $c$ be a color from the set of colors $\{1,\ldots,\sigma\}$. In row $c$ of $M$ we store a bit string where bit $i$ is \texttt{1} iff $v_i$ has color $c$. For each node $v_i$ we also store a bit string $A(i)$ where bit $j$ is \texttt{1} iff $v_j$ is an ancestor of $v_i$.

We construct this data structure from a packed colored tree as follows. Assume that the bit strings representing the node colorings form a $t \times \sigma$ matrix where row $i$ is the colorings of node $v_i$. We transpose this matrix to get $M$. To do this we partition the matrix into a $\frac t w \times \frac \sigma w$ matrix (assume w.l.o.g. that $w$ divides $t$ and $\sigma$), transpose each $w\times w$ submatrix as described in \cite{thorup2002randomized}, and transpose the $\frac t w \times \frac \sigma w$ matrix to get $M$. To compute the ancestor bit strings first set $A(root(\tree))=[\mathtt{0}]^t$. For all other nodes $v_i$, where $v_j$ is the parent of $v_i$, set $A(v_i)=A(v_j)\bor 2^j$.

%\paragraph{Query.} 
We answer a $\fca(v,c)$ as follows. Let $R=M[c]\band A(v)$. Now $R$ is a bit string representing the set of ancestors of $v$ with color $c$. Since the nodes have pre-order indices, the answer to the query is $v_i$, where $i$ is the index of the least significant set bit in $R$. %To find $i$ we scan $R$ from right to left to find the least significant non-zero word. In this word we determine the least significant set bit as shown in Lemma ...

To answer a $\lca(v,u,c)$ query we start by computing $R$ the same way as above. We then set the first $i-1$ bits of $R$ to \texttt{0}, where $i$ is the index of $u$. The answer to the query is the most significant set bit of $R$.% which is found similarly as in the above.

%\paragraph{Analysis.} 
The $\sigma \times t$ bit matrix $M$ can be packed in words and therefore uses $O(\frac{t\sigma}{w})$ space. The same is evident for the ancestor bit strings. Transposing a $w\times w$ matrix takes $O(w\log w)$ time, and since there are $\frac{t\sigma}{w^2}$ submatrices of this size in the color bit matrix, the total time spent for all submatrices is $O(\frac{t\sigma\log w}{w})$. Transposing the $\frac t w \times \frac \sigma w$ matrix takes $O(\frac{t\sigma}{w})$ time. Computing the ancestor bit strings clearly takes $O(\frac{t\sigma}{w})$ time.

The size of $R$ is $O(\frac t w)$, so finding the first non-zero word takes $O(\frac t w)$ time. Determining the least or most significant set bit of a word is done in $O(1)$ time. Thus, the query time for both a $\fca$ and a $\lca$ query is $O(\frac t w)$.

\begin{lemma}
The packed tree color problem can be solved using $O(t+\frac{t\sigma\log w}{w})$ preprocessing time,  $O(t+\frac{t\sigma}{w})$ space, and $O(\frac t w)$ query time.
\end{lemma}

\subsection{Combining the Solutions}
We now show how to combine the previously described solutions to get $\langle O(t+\frac{n\sigma}{w}), O(t+\frac{n\sigma}{w}), O(\log w) \rangle$ and $\langle O(t+\frac{t\sigma\log w}{w}), O(t+\frac{t\sigma}{w}), O(1) \rangle$ trade-offs. This is achieved by doing a cluster partioning of the tree.

%\paragraph{Data structure.} 
First we convert $\tree$ to a binary tree $\tree'$. Then we partition $\tree'$ into $O(\frac t w)$ clusters, i.e., each cluster has size $O(w)$. For each cluster $C$, where one boundary node is a leaf in the cluster and the other is the root of the cluster, we make a summary of the colors of the nodes on the path from the root to the leaf. The summary is stored in the macro tree node that corresponds to the leaf boundary node of $C$. Apply the $\langle O(t\sigma), O(t\sigma), O(1) \rangle$ solution to the macro tree, and apply either the $\langle O(\frac{t\sigma}{w}),O(\frac{t\sigma}{w}),O(\log t) \rangle$ solution or the $\langle O(\frac{t\sigma\log w}{w}),O(\frac{t\sigma}{w}),O(\frac t w) \rangle$ solution to each cluster using the original colors.

%\paragraph{Query.} 
Here is how we answer a $\fca(v,c)$ query. Let $C_v$ be the cluster containing $v$. First we ask for $\fca(v,c)$ in $C_v$. If the answer is a node in $C_v$, we are done. If it is undefined, we find the node $r$ in the macro tree corresponding to the root of $C_v$. We check if $r$ has color $c$ in the macro tree and otherwise ask for $w=\fca(r,c)$ in the macro tree. In the cluster $C_w$ having $w$ as a leaf boundary node we then check if $w$ has color $c$ and otherwise ask for $\fca(w,c)$ in $C_w$.

We answer a $\lca(u,v,c)$ query as follows. Assume that $u\neq v$ and let $C_u$ and $C_v$ be the clusters containing $u$ and $v$. If $C_u=C_v$ then the answer is $\lca(u,v,c)$ in the cluster containing $u$ and $v$. If $C_u\neq C_v$, let $w$ be the leaf boundary node of $C_u$ where $v\in \tree(w)$. We now proceed in three steps. First, we ask for $\lca(u,w,c)$ in $C_u$. If the query returns a node, this is also the answer to the $\lca(u,v,c)$ query. If the answer in the first step is undefined we ask for $z=\lca(w,root(C_v),c)$ in the macro tree to locate the highest cluster with a node with color $c$ between $u$ and $v$. The answer to the query is then $\lca(root(C_z),z,c)$ on $C_z$. If the first two steps fail, the answer to a query is $\lca(root(C_v),v,c)$.

%\paragraph{Analysis.} 
The cluster partition can be computed in linear time, and the cluster path summaries are computed in $O(\frac{t\sigma}{w})$ time. Since the macro tree has $O(\frac t w)$ nodes the preprocessing time and space to apply the $\langle O(t\sigma), O(t\sigma), O(1) \rangle$ solution becomes $O(\frac{t\sigma}{w})$. To answer a query we perform a constant number of $\fca$ and $\lca$ queries on the macro tree and clusters. Therefore the total time to perform queries on the macro tree is $O(1)$ time. To get $(i)$ we apply the $\langle O(t+\frac{t\sigma}{w}), O(t+\frac{t\sigma}{w}), O(\log t) \rangle$ solution to clusters. Since a cluster has size $O(w)$ we use a total of $O(\log w)$ time performing queries on clusters. To get $(ii)$ we apply the $\langle O(\frac{t\sigma\log w}{w}),O(\frac{t\sigma}{w}),O(\frac t w) \rangle$  solution to clusters. Again, since clusters have size $O(w)$ we use a total of $O(1)$ time performing queries on clusters. Preprocessing time and space for the cluster data structures follow because $\sum_{C\in CS}|C|=O(t)$.

\begin{theorem}\label{thm:PTC}
The packed tree color problem can be solved using $O(t+\frac{t\sigma}{w})$ space,
\begin{enumerate}
	\item[\textit{(i)}] $O(t+\frac{t\sigma}{w})$ preprocessing time, and $O(\log w)$ query time, or
	\item[\textit{(ii)}] $O(t+\frac{t\sigma}{w}\log w)$ preprocessing time, and $O(1)$ query time.
\end{enumerate}
\end{theorem}
% \begin{proof} The cluster partition can be computed in linear time, and the cluster path summaries are computed in $O(\frac{t\sigma}{w})$ time. Since the macro tree has $O(\frac t w)$ nodes the preprocessing time and space to apply the $\langle O(t\sigma), O(t\sigma), O(1) \rangle$ solution becomes $O(\frac{t\sigma}{w})$. To answer a query we perform a constant number of $\fca$ and $\lca$ queries on the macro tree and clusters. Therefore the total time to perform queries on the macro tree is $O(1)$ time. To get $(i)$ we apply the $\langle O(t+\frac{t\sigma}{w}), O(t+\frac{t\sigma}{w}), O(\log t) \rangle$ solution to clusters. Since a cluster has size $O(w)$ we use a total of $O(\log w)$ time performing queries on clusters. To get $(ii)$ we apply the $\langle O(\frac{t\sigma\log w}{w}),O(\frac{t\sigma}{w}),O(\frac t w) \rangle$  solution to clusters. Again, since clusters have size $O(w)$ we use a total of $O(1)$ time performing queries on clusters. Preprocessing time and space for the cluster data structures follow because $\sum_{C\in CS}|C|=O(t)$.\qed
% \end{proof}

\section{Labelled Successor Data Structure for SLPs}
The answer to a labelled successor \labsuc$(i,c)$ query on a string $\str$ is the index of the first occurrence of the character $c$ after position $i$ in $\str$. More formally, the answer to \labsuc$(i,c)$ is an index $j$ such that $S[j]=c$, $j>i$, and $S[k]\neq c$ for $k=i+1,\ldots ,j-1$. 

In this section we present a data structure that supports $\labsuc(i,c)$ queries on an SLP.  This is the first data structure dedicated to solving this problem on SLPs. Alternatively, we may build the random access data structure of \cite{bille2011random} and then answer an $\labsuc(i,c)$ query by doing a random access query for position $i$ followed by a linear scan to find the first occurrence of $c$. This yields a query time of $O(\log N+j-i)$ while using $O(n)$ space for the data structure. 

Our data structure combines the random access data structure of \cite{bille2011random} with a new way of navigating the SLP based on the characters of substrings. For the latter we will utilize our result for the packed tree color problem described in the previous section.

The basic idea is to store a bit string for each node $v\in \slp$ that summarizes which characters that are generated by $\str(v)$. We first seach for position $i$ in $\str$ and let $p$ be the unique path in $\slp$ defining $S[i]$. We then walk up $p$ until reaching a node $u$ where $\rightc(u)$ generates a string that contains $c$ and $\rightc(u)$ is not on $p$. Then we walk down from $\rightc(u)$ using the summaries to locate the leftmost terminal descending from $\rightc(u)$ that generates $c$. This algorithm requires $O(n+\frac{n\sigma}{w})$ space and $O(h)$ time to find $\labsuc(i,c)$.

To speed things up we fix a heavy path decomposition of the SLP to get a heavy forest and build the random access data structure of \cite{bille2011random}. Now $p$ is a sequence of entry and exit points in the trees of the heavy forest. When we walk up $p$ we enter a tree in an exit node and have to walk away from the root to the first node whose right child generates a string that contains $c$ before reaching the entry node. This is equivalent to a $\lca$ query. When we walk down to find $\labsuc(i,c)$ we enter a tree and have to walk towards the root to find either the first ancestor whose left child generates a string that contains $c$ or the highest ancestor whose right child generates $c$. This is equivalent to a $\fca$ and a $\lca$ query, respectively.

In the remainder of this section we give the details of the data structure.

\begin{theorem}\label{thm:LSDS}
There is a data structure supporting labelled successor (and predecessor) queries on a string of size $N$ over an alphabet of size $\sigma$ compressed by an SLP of size $n$ in the word RAM model with word size $w\geq \log N$ using $O(n+\frac{n\sigma}{w})$ space and

\begin{enumerate}
	\item[\textit{(i)}] $O(n+\frac{n\sigma}{w})$ preprocessing time, and $O(\log N\log w)$ query time, or
	\item[\textit{(ii)}] $O(n+\frac{n\sigma}{w}\log w)$ preprocessing time, and $O(\log N)$ query time.
\end{enumerate}
\end{theorem}

	%\subsection{Contents of the Data Structure} 
%\paragraph{Contents of the data structure.}
\begin{proof} We first apply the construction of \cite{bille2011random}, and let $\hf$ be the heavy forest obtained from the heavy path decomposition of $\slp$. For each node $v$ in $\slp$ with children $\leftc(v)$ and $\rightc(v)$ we store two bit strings $L(v)$ and $R(v)$ summarizing the characters in $S(\leftc(v))$ and $S(\rightc(v))$. If $v$ and $\leftc(v)$ are in the same tree in $\hf$ then $L(v)=[\mathtt{0}]^\sigma$ and similarly for $\rightc(v)$ and $R(v)$. For each tree in $\hf$ we build two data structures for the packed tree color problem.  One where the $L$ bit strings serve as colors and one where the $R$ bit strings serve as colors.

	%\subsection{Answering a Query}
	%\paragraph{Answering a query.}
We answer an $\labsuc(i,c)$ query as follows. First we access the character $S[i]$ using the random access data structure. We now have the entry and exit points of the heavy trees in $\hf$ on the unique path $p$ describing $S[i]$. Let $T_1,\ldots,T_k\in \hf$ be a sequence of trees on $p$ in the order they are visited when starting from the root and ending in the terminal generating $\str[i]$, and let $(v_1,u_1),\ldots,(v_k,u_k)$ be the entry and exit nodes for each tree in the sequence. Using the packed tree color data structure for the $R$ colors, we repeat $\lca(u_i,v_i,c)$ for $i=k$ down to some $j$ until $\lca(u_j,v_j,c)$ is not undefined. Let $w=\rightc(\lca(u_j,v_j,c))$. We now search for the first occurrence of $c$ in $S(w)$. Let $T_i$ be the tree in $\hf$ that contains the node $w$, then the search proceeds in three steps. First, we ask for $v=\fca(w,c)$ in $T_i$ in the data structure for $L$ colors and restart the search from $\leftc(v)$. If the query $\fca(w,c)$ is undefined we continue to the next step. In the second step we check if $root(T_i)$ generates $c$. If it does, we now have a unique set of entry and exit nodes in the trees of $\hf$ that constitutes a path to a terminal that generates the first $c$ after position $i$. The answer to the $\labsuc(i,c)$ query is the index of this $c$ which we retrieve using the random access data structure. Finally, if $root(T_i)$ does not generate $c$ we ask for $v=\lca(w, root(T_i),c)$ in $T_i$ in the data structure for $R$ colors, and restart the search from $\rightc(v)$.

%\paragraph{Analysis.}
The data structure uses $O(n+\frac{n\sigma}{w})$ space because the random access data structure uses linear space and the bit strings $L$ and $R$ use $O(\frac{n\sigma}{w})$ space. The random access data structure, including the heavy path decomposition, takes $O(n)$ time to compute and the $L$ and $R$ values are computed using $O(\frac{n\sigma}{w})$ OR operations in a bottom up fashion. Therefore, this part of the data structure is computed in $O(n+\frac{n\sigma}{w})$ time.

To get Theorem \ref{thm:LSDS} (i) we use the packed tree color data structure of Theorem \ref{thm:PTC} (i) for the trees in $\hf$ and likewise for (ii). Since the trees are disjoint, the preprocessing time and space becomes as in the Theorem \ref{thm:LSDS}.

For the query, we first do one random access query that takes $O(\log N)$ time, then we perform at most $\log N$ $\lca$ queries walking up the SLP and at most $2\log N$ $\fca$ and $\lca$ queries locating the labelled successor. Finally, retrieving the index also takes $O(\log N)$ time using the random access data structure.   \qed
\end{proof}

\section{Subsequence Matching}
We will now use the labelled successor data structure to obtain a subsequence matching algorithm for SLPs. Our algorithm is based on the folklore algorithm for subsequence matching  which works as follows (see also \cite{mannila1997discovery,das1997episode}). First we find the minimal prefix $\str[1..j]$ that contains $\pat$ as a subsequence. This is done by reading $\str$ left to right while searching for the characters of $\pat$ one at a time. We then find the minimal suffix $\str[i..j]$ of the prefix $\str[1..j]$ that contains $\pat$. Similarly, this is done by scanning the prefix right to left. Now $\str[i..j]$ is the first minimal occurrence of $\pat$. To find the next minimal occurrence we repeat this process for the suffix $\str[i+1..N]$. It can be shown that this algorithm finds all minimal occurrences of $\pat$ in $O(Nm)$ time.

By using our labelled successor data structure described in the previous section we speed up the procedure of finding some specific character of $\pat$. Assume we have matched $\pat[1..k]$ to $\str[1..j]$ such that $\pat[k]=\str[j]$. Instead of doing a linear scan of $\str[j+1..N]$ to find $\pat[k+1]$ we ask for the next occurrence of $\pat[k+1]$ using $\labsuc(j,\pat[k+1])$.

For each occurrence of $\pat$ we perform $O(m)$ labelled successor (and labelled predecessor) queries, and we also have to construct the data structures to support these. By applying the results of Theorem \ref{thm:LSDS} we get Theorem \ref{thm:SCSM}.

%
% ---- Bibliography ----
%
\bibliographystyle{abbrv}
\bibliography{references}

\end{document}